\documentclass[letterpaper,11pt]{article}

\usepackage{amssymb}
\usepackage{amsmath}
\usepackage{fullpage}
\usepackage[pagebackref=true]{hyperref}
\usepackage{algpseudocode}
\usepackage{algorithm}
\usepackage{amsthm}

\usepackage[T1]{fontenc}

\newcommand{\norm}[2]{\left \lVert#2\right \rVert_{#1}}
\newcommand{\eps}{\varepsilon}
\newcommand{\set}[1]{\left\{#1\right\}}
\newcommand{\Prb}[2]{\mathrm{Pr}_{#1}\left[#2\right]}
\newcommand{\Exp}[2]{\mathrm{E}_{#1}\left[#2\right]}
\newcommand{\Rbb}{\mathbb{R}}

\newcommand{\Mc}{\mathcal{M}}
\newcommand{\Bc}{\mathcal{B}}
\newcommand{\Tc}{\mathcal{T}}

\DeclareMathOperator{\supp}{supp}
\DeclareMathOperator{\argmin}{argmin}

\newtheorem{theorem}{Theorem}
\newtheorem{lemma}{Lemma}
\newtheorem{definition}{Definition}
\newtheorem{corollary}{Corollary}

\begin{document}
    \title{On Model-Based RIP-$1$ Matrices}
    \author{Piotr Indyk\thanks{MIT CSAIL, \texttt{indyk@mit.edu}} \and Ilya Razenshteyn\thanks{MIT CSAIL, \texttt{ilyaraz@mit.edu}}}
    \date{\vspace{-5ex}}
    \maketitle
    \begin{abstract}
        The Restricted Isometry Property (RIP) is a fundamental property of a matrix enabling sparse recovery~\cite{CRT06}. Informally, an $m \times n$ matrix satisfies RIP of order $k$ in the $\ell_p$ norm if $\|Ax\|_p \approx \|x\|_p$ for any vector $x$ that is $k$-sparse, i.e., that has at most $k$ non-zeros.  The minimal number of rows $m$ necessary for the property to hold has been extensively investigated, and tight bounds are known.
       Motivated by signal processing models,  a recent work of Baraniuk et al~\cite{BCDH10} has generalized this notion to the case where the support of $x$ must belong to a given {\em model}, i.e., a given family of supports.   This more general notion is much less understood, especially for norms other than $\ell_2$. 
       
       In this paper we present tight bounds for the model-based RIP property in the $\ell_1$ norm. Our bounds hold for the two most frequently investigated models: tree-sparsity and block-sparsity. We also show implications of our results to sparse recovery problems.
        \end{abstract}
    \section{Introduction}
    \label{intro}

In recent years, a new ``linear'' approach for obtaining a succinct
approximate representation of $n$-dimensional vectors (or signals) has
been discovered.  For any signal $x$, the representation is equal to
$Ax$, where $A$ is an $m \times n$ matrix, or possibly a random
variable chosen from some distribution over such matrices.  The vector
$Ax$ is often referred to as the {\em measurement vector} or {\em
  linear sketch} of $x$.  Although $m$ is typically much smaller than
$n$, the sketch $Ax$ often contains plenty of useful information about
the signal $x$.

A particularly useful and well-studied problem is that of {\em stable
  sparse recovery}. We say that a vector $x'$ is $k$-sparse if it has at most $k$ non-zero
coordinates.  The sparse recovery problem is typically defined as follows: for
some norm parameters $p$ and $q$ and an approximation factor $C>0$,
given $Ax$, recover an ``approximation'' vector $x^*$ such that
\begin{equation}
\label{e:lplq}
\norm{p}{x-x^*} \le C \min_{k\text{-sparse } x'}  \norm{q}{x-x'}
\end{equation}
(this inequality is often referred to as {\em $\ell_p/\ell_q$
  guarantee}).  
   Sparse recovery has a tremendous number of applications
in areas such as compressive sensing of
signals~\cite{CRT06,D06}, genetic
data acquisition and analysis and data stream
algorithms~\cite{M05,GI10}.
 
It is known~\cite{CRT06}
that there exist matrices $A$ and associated recovery algorithms that
produce approximations $x^*$ satisfying Equation~\eqref{e:lplq} with
$p=q=1$\footnote{In fact, one can prove a somewhat stronger guarantee, referred to as the $\ell_2/\ell_1$ guarantee.  }
 constant approximation factor $C$, and sketch length
\begin{equation}
\label{e:m}
m=O(k \log (n/k))
\end{equation}
This result was  proved by showing that there exist matrices $A$ with $m=O(k \log (n/k))$ rows that satisfy the {\em Restricted Isometry Property (RIP)}. Formally,  we say that $A$ is \emph{a $(k, \eps)$-RIP-$p$ matrix},
        if for every $x \in \Rbb^n$
        with at most $k$ non-zero coordinates we have
        $$
            (1 - \eps) \|x\|_p \leq \|Ax\|_p \leq (1 + \eps) \|x\|_p.
        $$
The proof of~\cite{CRT06} proceeds by showing that (i) there exist matrices with $m=O(k \log (n/k))$ rows that satisfy $(k, \eps)$-RIP-$2$ for some constant $\eps>0$ and (ii) for such matrices there exist a polynomial time recovery algorithm that given $Ax$ produces $x^*$ satisfying Equation~\ref{e:lplq}. Similar results were obtained for RIP-$1$ matrices~\cite{BGIKS08}. The latter matrices are closely connected to hashing-based streaming algorithms for heavy-hitter problems, see~\cite{GI10} for an overview.

It is known that the bound on the number of measurements in~Equation~\eqref{e:m} is
asymptotically optimal for some constant $C$ and $p=q=1$,
see~\cite{DIPW10} and~\cite{FPRU10} (building on~\cite{D06,GG84,G84,K77}).  
The necessity of the ``extra'' logarithmic factor multiplying $k$ is
quite unfortunate: the sketch length determines the ``compression
rate'', and for large $n$ any logarithmic factor can worsen that rate
tenfold.
Fortunately, a more careful modeling offers a way to overcome the aforementioned limitation. In particular, after decades of research in signal modeling, signal processing researchers know that not all supports (i.e., sets of non-zero coordinates) are equally common . For example, if a signal is a function of time, large coefficients of the signal tend to occur consecutively. This phenomenon can be exploited by searching for the best $k$-sparse approximation $x^*$ whose support belongs to a given ``model" family of supports $\Mc_k$ (i.e., $x^*$ is  $\Mc_k$-{\em sparse}). Formally, we seek $x^*$ such that
 \begin{equation}
 \label{e:model}
        \|x - x^*\|_p \leq C \cdot \min_{\begin{smallmatrix}\supp(x') \subseteq T \\ T \in \Mc_k\end{smallmatrix}} \|x - x'\|_q
 \end{equation}
 
    for some family $\Mc_k$ of $k$-subsets of $[n]$.
    Clearly, the original $k$-sparse recovery problem corresponds to the case, when $\Mc_k$ is a family of all $k$-subsets
    of $[n]$.

A prototypical example of a sparsity model is {\em block sparsity}~\cite{EM09}. Here the signal is divided into blocks of size $b$, and the non-zero coefficients belong to at most $k/b$ blocks. This model is particularly useful for bursty time signals, where the ``activity'' occurs during a limited time period, and is therefore contained in a few blocks.
Another example is {\em tree sparsity}~\cite{RCB01} which models the structure of wavelet coefficients. Here the non-zero coefficients form a rooted subtree in a full binary tree defined over the coordinates.\footnote{See Section~\ref{s:def} for formal definitions of the two models.}. For many such scenarios the size of the family $\Mc_k$ is much smaller than ${n \choose k}$, which in principle makes it possible to recover an approximation from fewer measurements. 

An elegant and very general model-based sparse recovery scheme was recently provided in a seminal work of Baraniuk et al~\cite{BCDH10}. The scheme has the property that, for any ``computationally tractable'' family of supports of ``small" size, it guarantees a near-optimal sketch length $m=O(k)$, i.e., without any logarithmic factors. This is achieved by showing the existence of matrices $A$ satisfying the {\em model-based} variant of RIP. Formally, we  say that $A$ satisfies
        $\eps$-$\Mc_k$-RIP-$p$  if
    \begin{equation}
    \label{e:modelrip}
           (1 - \eps)\|x\|_p \leq \|Ax\|_p \leq (1 + \eps)\|x\|_p
      \end{equation}
        for any $\Mc_k$-sparse  vector $x \in \Rbb^n$.

In~\cite{BCDH10} it was shown that there exist matrices with $m=O(k)$ rows that satisfy $\eps$-$\Mc_k$-RIP-$2$ as long as (i) either $\Mc_k$ is the block-sparse model and $b=\Omega(\log n)$ or (ii) $\Mc_k$ is the tree-sparse model. This property can be then used to give an efficient algorithm that, given $Ax$, finds $x^*$ satisfying a variant of the guarantee of Equation~\ref{e:model}.
 However,  the guarantees offered in~\cite{BCDH10}, when phrased in the $\ell_1/\ell_1$ framework, results in a super-constant approximation factor $C=\Theta(\sqrt{\log n})$~\cite{IP11}.  The question of whether this bound can be improved has attracted considerable attention in signal processing and streaming communities. 
In particular, one of the problems\footnote{See Question 15: Sparse Recovery for Tree Models. The question was posed by the first author.} listed in the Bertinoro workshop open problem list~\cite{B11} asks whether there exist matrices $A$ with $m=O(k)$ rows  that provide the $\ell_1/\ell_1$ guarantee for the tree-sparse model with some constant approximation factor $C$.

\paragraph{Our results} In this paper we make a substantial progress on this question. In particular:

\begin{enumerate}
\item For both block-sparse and tree-sparse models, we show that there exist $m\times n$ matrices $A$ that provide the $\ell_1/\ell_1$ guarantee for some constant approximation factor $C$, such that the number of measurements improves over the bound of Equation~\ref{e:m} for a wide range of parameters $k$ and $b$. In particular we show that for the block-sparse model we can achieve $m=O(k \cdot (1 + \log_b (n / k)))$. 
This improves over the $O(k \log(n/k))$ bound of Equation~\ref{e:m}  for any $b = \omega(1)$.
In particular, if $b=n^{\Omega(1)}$, we obtain $m=O(k)$. 
For the tree-sparse model we achieve $m=O(k \log(n/k)/\log \log (n/k))$ as long as $k=\omega(\log n)$. This also improves over the $O(k \log(n/k))$ bound of Equation~\ref{e:m}.

 We note, however,  that our results are not accompanied by efficient recovery algorithms. Instead, we show the existence of model-based RIP-1 matrices with the given number of rows. This implies that $Ax$ contains enough information to recover the desired approximation $x^*$ (see Appendix~\ref{s:imp} for more details). 
\item We complement the aforementioned results by showing that the measurement bounds achievable for a matrix satisfying block-sparse or tree-sparse RIP-1 property cannot be improved (i.e., our upper bounds are tight). This provides strong evidence that the number of measurements required for sparse recovery itself cannot be $O(k)$. 
\end{enumerate}

Our results show a significant difference between the model-based RIP-$1$ and RIP-$2$ matrices. For the $\ell_2$ norm, the original paper~\cite{BCDH10} shows that the number of measurements  is fully determined by the cardinality of the model. Specifically,  their proof proceeds by applying the union bound over all elements
of $\Mc_k$ on top of the Johnson--Lindenstrauss-type concentration inequality.
This leads to a measurement bound of $m = O(k + \log |\Mc_k|)$, which is $O(k)$ for the tree-sparse or block-sparse models.
In contrast, in case of the $\ell_1$ norm we obtain an upper bound with a very different dependence on
$|\Mc_k|$: namely, if $|\Mc_k| \geq n / k$, then we are able to achieve
$$
    m = O\left(\frac{k \cdot \log \frac{n \log(n / k)}{\log |\Mc_k|}}{\log \frac{k \log (n / k)}{\log |\Mc_k|}}\right).
$$
Moreover, our lower bounds show that this bound is tight for the block- and the tree-sparse cases.

\paragraph{Our techniques} Our lower bounds are obtained by relating RIP-1 matrices to novel combinatorial/geometric structures we call  {\em generalized expanders}. Specifically, it is known~\cite{BGIKS08} that any {\em binary} 0-1 matrix $A$ that satisfies $(k, \eps)$-RIP-$1$ is an adjacency matrix of an unbalanced $(k, \eps)$-expander (see Section~\ref{s:def} for the formal definition). The notion of a generalized expander can be viewed as extending the notion of expansion to matrices that are not binary. Formally, we define it as follows.

    \begin{definition}[Generalized expander]
        \label{def_gen_expander}
        Let $A$ be an $m \times n$ real matrix. We say that $A$ is \emph{a generalized $(k, \eps)$-expander}, if 
        all $A$'s columns have $\ell_1$-norm at most $1 + \eps$, and for
        every $S \subseteq [n]$ with $|S| \leq k$ we have
        $$
            \sum_{i \in [m]} \max_{j \in S} |a_{ij}| \geq |S| \cdot (1 - \eps).
        $$
    \end{definition}

Observe that the notion coincides with the standard notion of expansion for binary 0-1 matrices (after a proper scaling).

In this paper we show that any (not necessarily binary) RIP-1 matrix is also a generalized expander. We then use this fact to show that any RIP-1 matrix can be sparsified by replacing most of its entries by 0. This in turn lets us use counting arguments to lower bound  the number of rows of such matrix. 

Our upper bounds are obtained by constructing low-degree expander-like graphs. However, we only require that the expansion holds for the sets from the given model $\Mc_k$. This allows us to reduce the number of the right nodes of the graph, which corresponds to reducing the number of rows in its adjacency matrix.

    \section{Definitions}
    \label{s:def}

    In this section we provide the definitions we will use throughout the text.

    \begin{definition}[Expander]
        \label{def_expander}
        Let $G = (U, V, E)$ with $|U| = n$, $|V| = m$, $E \subseteq U \times V$ be a bipartite graph
        such that all vertices from $U$ have the same degree $d$. Then we say that $G$ is \emph{a $(k, \eps)$-expander},
        if for every $S \subseteq U$ with $|S| \leq k$ we have
        $$
            |\set{v \in V \mid \exists u \in S \; (u, v) \in E}| \geq (1 - \eps) d |S|.
        $$
    \end{definition}

    \begin{definition}[Model]
        Let us call any non-empty subset
        $$
            \Mc_k \subseteq \Sigma_k = \set{A \subseteq [n] \mid |A| = k}
        $$
        \emph{a model}.
    \end{definition}

    In particular, $\Sigma_k$ is a model as well.

    \begin{definition}[Block-sparse model]
        Suppose that $b, k \in [n]$. Moreover, $b$ divides both $k$ and $n$.
        Let us partition our universe $[n]$ into $n / b$ disjoint blocks $B_1$, $B_2$, \ldots, $B_{n/b}$ of size $b$.
        We consider the following \emph{block-sparse model}: $\Bc_{k, b}$ consists of all unions of $k / b$ blocks.
    \end{definition}

    \begin{definition}[Tree-sparse model]
        Suppose that $k \in [n]$ and $n = 2^{h+1} - 1$, where $h$ is a non-negative integer.
        Let us identify the elements of $[n]$ with the vertices of a full
        binary tree of depth $h$.
        Then, \emph{tree-sparse model} $\Tc_k$ consists of all subtrees of size $k$ that contain the root of the full binary tree.
    \end{definition}

    \begin{definition}[Model-sparse vector/set]
        Let $\Mc_k \subseteq \Sigma_k$ be any model.
        We say that a set $S \subseteq [n]$ is \emph{$\Mc_k$-sparse}, if $S$ lies within a set from $\Mc_k$.
        Moreover, let us call a vector $x \in \Rbb^n$ \emph{$\Mc_k$-sparse}, if its support is a $\Mc_k$-sparse set.
    \end{definition}

    It is straightforward to generalize the notions of RIP-$p$ matrix, expanders and generalized expanders to the case of $\Mc_k$-sparse vectors
    and sets.
    Let us call the corresponding objects \emph{$\eps$-$\Mc_k$-RIP-$p$ matrix}, \emph{$\eps$-$\Mc_k$-expander} and
    \emph{generalized $\eps$-$\Mc_k$-expander}, respectively.
    Clearly, the initial definitions correspond to the case of $\Sigma_k$-sparse vectors and sets.

    Our two main objects of interest are $\Bc_{k, b}$- and $\Tc_k$-RIP-$1$ matrices.

    \section{Sparsification of RIP-$1$ matrices}

    In this section we show that any $n \times m$ matrix, which is $(k, \eps)$-RIP-$1$, can be sparsified after removing
    $(1 - \Omega(1)) n$ columns (Theorem~\ref{sparsification}). Then we state an obvious generalization of this fact
    (Theorem~\ref{model_sparsification}), which will be useful for proving lower bounds on the number of rows for $\Bc_{k, b}$- and $\Tc_k$-RIP-$1$ matrices.

    \begin{theorem}
        \label{sparsification}
        Let $A$ be any $m \times n$ matrix, which is $(k, \eps)$-RIP-$1$. Then there exists an $m \times \Omega(n)$ matrix
        $B$ which is $(k, O(\eps))$-RIP-$1$, has at most $O(m / k)$ non-zero entries per column and can be obtained from $A$
        by removing some columns and then setting some entries to zero.
    \end{theorem}

    We prove this theorem via the sequence of lemmas. First we prove that for every matrix $A$ there exists a $\pm 1$-vector
    $x$ such that $\|Ax\|_1$ is small.

    \begin{lemma}
        Let $A$ be any $m \times k$ matrix. 
        Then there exists a vector $x \in \set{-1, 1}^k$ such that
        \begin{equation}
            \label{l1_l2}
            \|Ax\|_1 \leq \sum_{i \in [m]} \left(\sum_{j \in [k]} a_{ij}^2\right)^{1/2}. 
        \end{equation}
    \end{lemma}
    \begin{proof}
        Let us use a probabilistic argument. Namely, let us sample all coordinates $x_i$ independently and uniformly at random from 
        $\{-1,1\}$. Then
        \begin{multline*}
            \Exp{}{\|Ax\|_1} = \sum_{i \in [m]} \Exp{}{\left|\sum_{j \in [k]} a_{ij} x_j\right|}
            \leq \sum_{i \in [m]} \left(\Exp{}{\left(\sum_{j \in [k]} a_{ij} x_j\right)^2}\right)^{1/2} =\\=
            \sum_{i \in [m]} \left(\Exp{}{\sum_{j \in [k]} a_{ij}^2 x_j^2}\right)^{1/2} =
            \sum_{i \in [m]} \left(\sum_{j \in [k]} a_{ij}^2\right)^{1/2}.
        \end{multline*}
        Thus, there exists a vector $x \in \set{-1, 1}^k$ that satisfies~(\ref{l1_l2}).
    \end{proof}

    As a trivial corollary we have the following statement.

    \begin{corollary}
        \label{l1_isometry_l2}
        Let $A$ be any $m \times k$ matrix that preserves (up to $1 \pm \eps$) $\ell_1$-norms of \emph{all} vectors.
        Then
        $$
            \sum_{i \in [m]} \left(\sum_{j \in [k]} a_{ij}^2\right)^{1/2} \geq (1 - \eps) k.
        $$
    \end{corollary}

    The next lemma shows that every $(k, \eps)$-RIP-$1$ matrix is a generalized $(k, O(\eps))$-expander.
    This is a generalization of a theorem from~\cite{BGIKS08}.

    \begin{lemma}
        \label{RIP_expander}
        Let $A$ be any $m \times n$ matrix, which is $(k, \eps)$-RIP-$1$. Then, $A$ is a generalized $(k, 3 \eps)$-expander.
    \end{lemma}
    \begin{proof}
        For the proof we need the following lemma.
        \begin{lemma}
            \label{l1_l2_linfty}
            For any $y \in \Rbb^k$ 
            \begin{equation}
                \label{l1_l2_linfty_eq}
                \|y\|_1 - \|y\|_{\infty} \leq \left(1 + \frac{1}{\sqrt{2}}\right)(\|y\|_1 - \|y\|_{2}).
            \end{equation}
        \end{lemma}
        \begin{proof}
            Clearly, if $y = 0$, then the desired inequality is trivial.
            Otherwise, by homogenity we can assume that $\|y\|_1 = 1$.
            If $\|y\|_{\infty} = 1$, then $\|y\|_2 = 1$, and both sides of~(\ref{l1_l2_linfty_eq})
            are equal to zero.
            So, we can assume that $\|y\|_{\infty} < 1$.
            Suppose that $\|y\|_{\infty} = t$ for some $t \in (0; 1)$.
            If $1 / n > t \geq 1 / (n + 1)$ (thus, $n = \lceil 1/t - 1 \rceil$) for some positive integer $n$, then, clearly,
            $\|y\|_2 \leq \sqrt{nt^2 + (1 - nt)^2}$.
            One can check using elementary analysis that for every $t \in (0; 1)$
            $$
                \frac{1 - \|y\|_{\infty}}{1 - \|y\|_2} \leq 
                \frac{1 - t}{1 - \sqrt{\left\lceil \frac{1}{t} - 1\right\rceil t^2
                    + \left(1 - \left \lceil \frac{1}{t} - 1\right \rceil t\right)^2
                }}
                \leq 1 + \frac{1}{\sqrt{2}} 
            $$
            (equality is attained on $t = 1/2$). This concludes the proof.
        \end{proof}

        Let $S \subseteq [n]$ be any subset of size at most $k$.
        For any $i \in [m]$ let us denote $y_i = (a_{ij})_{j \in S} \in \Rbb^S$.

        We have
    \begin{eqnarray*}
            \sum_{i \in [m]} \|y_i\|_{\infty} & \geq & 
            \left(1 + \frac{1}{\sqrt{2}} \right)
            \sum_{i \in [m]} \|y_i\|_2 - \frac{1}{\sqrt{2}} \cdot \sum_{i \in [m]} \|y_i\|_1   \mbox{\ \ \ \ (by Lemma~\ref{l1_l2_linfty})} \\
            & \geq &  \left(1 + \frac{1}{\sqrt{2}}\right) (1 - \eps)|S| - \frac{1}{\sqrt{2}} \cdot (1 + \eps) |S|   \mbox{\ \ \ \ (by Corollary~\ref{l1_isometry_l2} and RIP-$1$)} \\
            & = & (1 - (1 + \sqrt{2}) \eps)|S|.
        \end{eqnarray*}
        So, $A$ is a generalized $(k, (1 + \sqrt{2}) \eps)$-expander.
        Since $1 + \sqrt{2} < 3$, this concludes the proof.

    \end{proof}

    Finally, we prove Theorem~\ref{sparsification}.

    \begin{proof}[Proof of Theorem~\ref{sparsification}]
        By Lemma~\ref{RIP_expander} $A$ is a generalized $(k, 3 \eps)$-expander.
        Let us partition $[n]$ into $n / k$ disjoint sets of size $k$ arbitrarily: $[n] = S_1 \cup S_2 \cup \ldots
        \cup S_{n / k}$.
        Now for every $i \in [m]$ and every $S_t$ let us zero out all the entries
        $a_{ij}$ for $j \in S_t$ except one with the largest absolute value.
          Let $A'$ be the resulting matrix.
        Since $A$ is a generalized $(k, 3 \eps)$-expander, we know that the (vector) $\ell_1$ norm of the difference $A-A'$  is at most $3 \eps n$.
        Thus, each column of $A-A'$ has the $\ell_1$ norm of at most  $3 \eps$ \emph{on the average}.
        The number of non-zero entries in $A'$is at most $mn / k$, so a column has at most $m / k$ non-zero
        entries \emph{on the average}. Thus, by Markov inequality there is a set of $n / 3$ columns such that we have moved
        each of them by at most $9 \eps$ and each of them contains at most $3m / k$ non-zero entries.
        We define a matrix $B$ that consists of these columns.
        Since we have modified each of these columns by at most $9 \eps$ and $A$ is $(k, \eps)$-RIP-$1$ we have that
        $B$ is $(k, 10 \eps)$-RIP-$1$.
    \end{proof}

    The following theorem is a straightforward generalization of Theorem~\ref{sparsification}. It can be proved via literally the same argument.

    \begin{theorem}
        \label{model_sparsification}
        Suppose that a model $\Mc_k \subseteq \Sigma_k$ has the following properties:
        \begin{itemize}
            \item for some $l \leq k$ all sets from $\Sigma_l$ are $\Mc_k$-sparse;
            \item there exists a partition of an $\Omega(1)$-fraction of $[n]$ into disjoint subsets of size
            $\Omega(k)$ such that each of these subsets is $\Mc_k$-sparse.
        \end{itemize}
        Then if $A$ is an $m \times n$ matrix which is
        $\eps$-$\Mc_k$-RIP-$1$ for some sufficiently small $\eps >0$, there exists an $m \times \Omega(n)$ matrix $B$ which is
        $(l, O(\eps))$-RIP-$1$, has at most $O(m / k)$ non-zero entries per column and can be obtained from $A$
        by removing some columns and then setting some entries to zero.
    \end{theorem}

    \section{Lower bounds for model-based RIP-$1$ matrices}

    In this section we prove lower bounds on the number of rows for
    $\Bc_{k,b}$- and $\Tc_k$-RIP-$1$ matrices.

    This is done using the following general theorem.

    \begin{theorem}
        \label{general_lower_bound}
        If a model $\Mc_k \subseteq \Sigma_k$ satisfies the statement of Theorem~\ref{model_sparsification}
        and $A$ is an $m \times n$ matrix which is $\eps$-$\Mc_k$-RIP-$1$ for some sufficiently small  $\eps >0$, then
        $$
            m = \Omega\left(k \cdot \frac{\log(n / k)}{\log(k / l)}\right).
        $$
    \end{theorem}

    The proof is a combination of Theorem~\ref{model_sparsification} and a counting argument similar to one used in~\cite{N10}.

    First, we need the following standard geometric fact.

    \begin{theorem}
        \label{volume_bound}
        Let $v_1, v_2, \ldots, v_n \in \Rbb^d$ be a set of $d$-dimensional vectors such that
        \begin{itemize}
            \item for every $i \in [n]$ we have $\|v_i\|_1 \leq 1.1$;
            \item for every $i \ne j \in [n]$ we have $\|v_i - v_j\|_1 \geq 0.9$.
        \end{itemize}
        Then, $n \leq 4^d$.
    \end{theorem}
    \begin{proof}
        Denote $B(x, r)$ the ball in $\ell_1$-metric with center $x$ and
        radius $r$.
        Consider the balls $B_i = B(v_i, 0.45)$. On the one hand, these balls
        are disjoint, on the other hand, they lie within $B(0, 1.55)$. Thus, if
        we consider the balls' volumes we see that
        $$
            n \leq \left(\frac{1.55}{0.45}\right)^d < 4^d.
        $$
    \end{proof}

    The next theorem shows a tradeoff between $m$ and column sparsity for any RIP-$1$ matrix.
    Its variant was proved in~\cite{N10}, but we present here the proof for the sake of completeness.

    \begin{theorem}[\cite{N10}]
        \label{sparsity_m_tradeoff}
        Let $A$ be an $m \times n$ matrix, which is $(k, \eps)$-RIP-$1$ for some sufficiently small $\eps>0$.
        Moreover, suppose that every column of $A$
        has at most $s$ non-zero entries.
        Then 
        $$
            s \log \left(\frac{m}{sk}\right) = \Omega\left(\log\left(\frac{n}{k}\right)\right).
        $$
    \end{theorem}
    \begin{proof}
        We need a lemma from~\cite{N10}, which is proved by a standard probabilistic argument.
        \begin{lemma}
            There exists a set $X \subseteq \Rbb^n$ of $k / 2$-sparse vectors such that
            \begin{itemize}
                \item $\log |X| = \Omega(k \log(n / k))$;
                \item every vector from $X$ has a unit $\ell_1$-norm;
                \item all pairwise $\ell_1$-distances between the elements of $X$ are at least $1$.
            \end{itemize}
        \end{lemma}

        Now let us see how $A$ acts on the elements of $X$.
        Clearly, for every $x \in X$ the vector $Ax$ is $sk$-sparse.
        By pigeonhole principle we have that for some $S \subseteq [m]$ with $|S| \leq sk$
        there exists a subset $X' \subseteq X$ with
        \begin{equation}
            \label{php}
            |X'| \geq \frac{|X|}{\binom{m}{sk}}
        \end{equation}
        such that for every $x \in X'$ the support of $Ax$ lies within $S$.

        On the other hand, since $A$ is $(k, \eps)$-RIP-$1$ one can easily see that the
        set $\set{Ax}_{x \in X'}$ (which lies in the $sk$-dimensional subspace) has the following properties:
        \begin{itemize}
            \item
                every vector from the set has $\ell_1$-norm at most $1 + \eps$;
            \item
                all pairwise distances are at least $1 - \eps$.
        \end{itemize}
        Since this set lies in the $sk$-dimensional subspace by Theorem~\ref{volume_bound} its cardinality is
        bounded by $4^{sk}$ (provided that $\eps$ is sufficiently small).
        Thus, we have by plugging this bound into~(\ref{php})
        $$
            \frac{2^{\Omega(k \log(n / k))}}{\binom{m}{sk}} \leq 4^{sk}. 
        $$
        Now by using a standard estimate $\binom{m}{sk} \leq 2^{O(sk \log(m / sk))}$ we have the desired statement.
    \end{proof}

    Now we can finish the proof of Theorem~\ref{general_lower_bound}.

    \begin{proof}[Proof of Theorem~\ref{general_lower_bound}]
        By Theorem~\ref{model_sparsification} we can get an $m \times \Omega(n)$ matrix $A$
        with column sparsity $s = O(m / k)$ and which is $(l, O(\eps))$-RIP-$1$. 
        Then applying Theorem~\ref{sparsity_m_tradeoff} we have
        $s \log(m / sl) = \Omega(\log(n / l))$. Since, $s = O(m / k)$ we get
        the desired bound
        $$
            m = \Omega\left(k \cdot \frac{\log(n / k)}{\log(k / l)}\right).
        $$
    \end{proof}

    Next we apply Theorem~\ref{general_lower_bound} to $\Bc_{k, b}$- and $\Tc_k$-RIP-$1$ matrices.

    \begin{theorem}
        For any $k \geq 2 b$ and sufficiently small $\eps>0$ if $A$ is an $m \times n$ matrix which is $\eps$-$\Bc_{k, b}$-RIP-$1$,
        then $m = \Omega(k \cdot (1 + \log_b (n / k)))$.
    \end{theorem}
    \begin{proof}
        Clearly, $\Bc_{k, b}$ satisfies the conditions of Theorem~\ref{model_sparsification}
        for $l = k / b$.
        Thus, by Theorem~\ref{general_lower_bound} we have
        $$
            m = \Omega\left(k \cdot \log_b(n / k)\right).
        $$
        At the same time, the lower bound $\Omega(k)$ is obvious. Combining them together
        we get the desired result.
    \end{proof}

    \begin{theorem}
        Let $A$ be an $m \times n$ matrix which is $\eps$-$\Tc_k$-RIP-$1$.
        Then, if $\eps$ is sufficiently small  and $k = \omega(\log n)$,
        $$
            m = \Omega\left(k \cdot \frac{\log(n / k)}{\log \log (n / k)}\right). 
        $$
    \end{theorem}
    \begin{proof}
        The next Lemma shows that for any $k = \omega(\log n)$ the model $\Tc_k$ satisfies the first condition of Theorem~\ref{model_sparsification}
        with $l = \Omega(k / \log (n / k))$. 
        \begin{lemma}
            Let $S \subseteq [n]$ be a subset of the full binary tree. Then there exists a subtree that contains both $S$ and the root
            with at most $O(|S| \log(n / |S|))$ vertices.
        \end{lemma}
        \begin{proof}
            Let $T$ be a subtree that consists of $\log |S|$ levels of the full binary tree that are closest to the root.
            Let $T'$ be a subtree that is a union of $T$ and paths from the root to all the elements of $|S|$.
            It is not hard to see that $|T' \setminus T| \leq |S| \log (n / |S|)$.
            As a result we get
            $$
                |T'| \leq |T| + |S| \log (n / |S|) \leq O(|S| \log (n / |S|)).
            $$
        \end{proof}
        The second condition of Theorem~\ref{model_sparsification} is satisfied as well (here we use that $k = \omega(\log n)$).
        Thus, applying Theorem~\ref{general_lower_bound} we have
        $$
            m = \Omega\left(k \cdot \frac{\log(n / k)}{\log \log(n / k)}\right).
        $$
    \end{proof}

    \section{Upper bounds for model-based RIP-$1$ matrices}
    \label{upper_bounds}

    In this section we complement the lower bounds by upper bounds.

    We use the following obvious modification of a theorem from~\cite{BGIKS08}.
    \begin{theorem}[\cite{BGIKS08}]
        \label{expanders_rip}
        If a graph $G = (U, V, E)$ is an $\eps$-$\Mc_k$-expander for some
        model $\Mc_k \subseteq \Sigma_k$, then the normalized (by a factor of
        $d$, where $d$ is the degree of all vertices from $U$)
        adjacency matrix of $G$ (which size is $|V| \times |U|$)
        is an $O(\eps)$-$\Mc_k$-RIP-$1$ matrix.
    \end{theorem}

    Thus, it is sufficient to build $\Mc_k$-expanders
    with as small $m$ as possible.
    We use the standard probabilistic argument to show the existence of such
    graphs. Namely, for every vertex $u \in U$ we sample 
    a subset of $[m]$ of size $d$ ($d$ and $m$ have to be carefully
    chosen).
    Then, we connect $u$ and all the vertices from this subset.
    All sets we sample are uniform (among all $d$-subsets of $[m]$)
    and independent.

    We use the following tail inequality, which can be proved using
    Chernoff bound (and whose slight
    variant is proved and used in~\cite{BMRV02}).

    \begin{lemma}[\cite{BMRV02}]
        \label{concentration}
        There exist constants $C > 1$ and $\delta > 0$ such that, whenever
        $m \geq C dt / \eps$, one has for any $T \subseteq U$ with $|T| = t$
        $$
            \Prb{}{|\set{v \in V \mid \exists u \in T \; (u, v) \in E}|
            < (1 - \eps) dt} \leq \left(\delta \cdot
                    \frac{\eps m}{dt}\right)^{-\eps dt}
        $$
    \end{lemma}

    For the proof see Appendix (Section~\ref{concentration_proof}).

    For a model $\Mc_k \subseteq \Sigma_k$ denote $\#(\Mc_k, t)$ the number
    of $\Mc_k$-sparse sets of size $t$ (for $t \in [k]$).
    We use the following simple estimate.

    \begin{lemma}
        $$
            \forall t \in [k] \quad \#(\Mc_k, t) \leq
            \min\set{|\Mc_k| \cdot \binom{k}{t}, \binom{n}{t}}.
        $$
    \end{lemma}

    Now if we combine this Lemma with the standard estimate $\binom{u}{v} \leq
    (eu / v)^v$, we get the following bound. 

    \begin{equation}
        \label{block_estimates}
            \forall t \in [k] \quad \#(\Mc_k, t) \leq \min \set{|\Mc_k| \cdot \left(\frac{ek}{t}\right)^t,
                \left(\frac{en}{t}\right)^t} 
    \end{equation}

    Now we combine these estimates, Lemma~\ref{concentration}, Theorem~\ref{expanders_rip} and the union bound
    to get upper bounds for $m$ for $\Mc_k$-RIP-$1$ matrices.

    \begin{theorem}
        \label{model_upper}
        For any $0 < \eps < 1/2$ and any model $\Mc_k$
        that is of size at least $n / k$
        there exists an $\eps$-$\Mc_k$-RIP-$1$ matrix with
        $$
        m = O\left(\frac{k}{\eps^2} \cdot \frac{\log(n / l)}{\log(k / l)}\right),
        $$
        where
        $$
            l = \frac{\log |\Mc_k|}{\log(n / k)}.
        $$
    \end{theorem}

The proof is in the appendix. We note that the condition $|\Mc_k| \geq n / k$ is needed to make sure that
$l \geq 1$.

Now we derive corollaries for the block- and tree-sparse cases.

\begin{corollary}
    \label{block_upper}
    For any $0 < \eps < 1/2$ and any $1 \leq b \leq k \leq n$
    there exists an $\eps$-$\Bc_{k,b}$-RIP-$1$ matrix with
    $$
        m = O\left(\frac{k \cdot (1 + \log_b(n/k))}{\eps^2}\right)
    $$
    rows.
\end{corollary}
\begin{proof}
    It is easy to see that
    $$
        |\Mc_k| = \binom{n/b}{k/b} \geq n/k.
    $$
    Thus, we can apply Theorem~\ref{model_upper} with
    $
        l = k/b
    $
    and get the required bound.
\end{proof}

\begin{corollary}
    \label{tree_upper}
    For any $0 < \eps < 1/2$ and any $k = \omega(\log n)$
    there exists an $\eps$-$\Tc_k$-RIP-$1$ matrix with
    $$
        m = O\left(\frac{k}{\eps^2}\cdot\frac{\log(n / k)}{\log \log(n/k)}\right)
    $$
    rows.
\end{corollary}
\begin{proof}
    Using a simple estimate on Catalan numbers we can see
    that $|\Tc_k| \leq 4^k$. Since $k = \omega(\log n)$,
    we are in position to apply Theorem~\ref{model_upper} and
    get the required bound.
\end{proof}
 
    \section{Acknowledgments}
        This work was supported in part by NSF CCF-1065125 award, by MADALGO (a Center
        of the Danish National Research Foundation), by Packard Foundation and by Akamai Presidential Fellowship. 

    The second author would like to thank Jelani Nelson for useful discussions.
    We also thank Ludwig Schmidt who pointed out an error in an
    earlier version of this paper.

    \bibliographystyle{alpha}
    \bibliography{../../../bibtex/ir}

\appendix

    \section{RIP-$1$ yields sparse recovery}
    \label{s:imp}

    In this section we show improved upper bounds on the number of measurements needed to recover good block- or tree-sparse approximations with
    $\ell_1 / \ell_1$ guarantee and constant approximation factor. This result is folklore, but we include it for completeness.

    Suppose that $\Mc_k \subseteq \Sigma_k$ is some model. We say that an $m \times n$ matrix $A$ is $\eps$-$\Mc_k^{(2)}$-RIP-$1$, if for every
    $x \in \Rbb^n$ such that $\supp x \subseteq S_1 \cup S_2$ for some $\Mc_k$-sparse sets $S_1$ and $S_2$ one has
    $$
        (1 - \eps) \|x\|_1 \leq \|Ax\|_1 \leq (1 + \eps) \|x\|_1.
    $$

    Let $A$ be \emph{any} $\eps$-$\Mc_k^{(2)}$-RIP-$1$ matrix for a sufficiently small $\eps$ such that $\|A\|_1 \leq 1 + \eps$.
    Algorithm~\ref{recovery_algorithm} (whose running time is exponential in $n$) given $y = Ax$ for some $x \in \Rbb^n$ recovers 
    a vector $x^* \in \Rbb^n$ such that
    \begin{equation}
        \label{good_approximation}
        \|x - x^*\|_1 \leq (3 + O(\eps)) \cdot \min_{\mbox{$x'$ is $\Mc_k$-sparse}} \|x - x'\|_1.
    \end{equation}
    \begin{algorithm}
        \caption{Model-based sparse recovery}
        \label{recovery_algorithm}
        \begin{algorithmic}
            \Require $y = Ax$ for some $x \in \Rbb^n$
            \Ensure a good $\Mc_k$-approximation $x^*$ of $x$
            \State $x^* \gets 0$
            \For{$S \subseteq [n]$ is an $\Mc_k$-sparse set}
                \State $\tilde{x} \gets \argmin_{\supp x' \subseteq S} \|y - Ax'\|_1$
                \If{$\|y - A\tilde{x}\|_1 \leq \|y - Ax^*\|_1$}
                    \State $x^* \gets \tilde{x}$
                \EndIf
            \EndFor
        \end{algorithmic}
    \end{algorithm} 
    Note that the optimization problem within the for-loop can be easily reduced to a linear program.
    Now let us prove that the resulting vector $x^*$ satisfies~(\ref{good_approximation}).
    Denote $x_{\Mc_k}$ an $\Mc_k$-sparse vector that minimizes $\|x - x_{\Mc_k}\|_1$.
    Now we have
    \begin{multline*}
        \|x - x^*\|_1 \leq \|x - x_{\Mc_k}\|_1 + \|x_{\Mc_k} - x^*\|_1 \leq
        \|x - x_{\Mc_k}\|_1 + (1 + O(\eps)) \|A(x_{\Mc_k} - x^*)\|_1 \leq \\
        \leq \|x - x_{\Mc_k}\|_1 + (1 + O(\eps)) (\|A(x - x_{\Mc_k})\|_1 + \|A(x - x^*)\|_1) \leq \\
        \leq (2 + O(\eps)) \|x - x_{\Mc_k}\|_1 + (1 + O(\eps)) \|A(x - x^*)\|_1 \leq \\
        \leq (2 + O(\eps)) \|x - x_{\Mc_k}\|_1 + (1 + O(\eps)) \|A(x - x_{\Mc_k})\|_1 
        \leq (3 + O(\eps)) \|x - x_{\Mc_k}\|_1.
    \end{multline*}
    The second inequality is true, since $A$ is $\eps$-$\Mc_k^{(2)}$-RIP-$1$ and both $x_{\Mc_k}$ and $x^*$ are $\Mc_k$-sparse.
    The fourth inequality is true, because  $\|A\|_1 \leq 1 + \eps$.
    The fith inequality is due to the construction of the algorithm: clearly, $\|A(x - x^*)\|_1 \leq \|A(x - x_{\Mc_k})\|_1$.

    It is immediate to see that any $\eps$-$\Bc_{2k, b}$-RIP-$1$ matrix is $\eps$-$\Bc_{k, b}^{(2)}$-RIP-$1$.
    Similarly, any $\eps$-$\Tc_{2k}$-RIP-$1$ matrix is $\eps$-$\Tc_k^{(2)}$-RIP-$1$.
    Moreover, since all the singletons are both block- and tree-sparse, we have that these matrices have $\ell_1$-norm at most $1 + \eps$. 
    Thus, plugging Corollaries~\ref{block_upper} and~\ref{tree_upper}
    we get the following result.

    \begin{theorem}
        The problem of model-based stable sparse recovery with $\ell_1 / \ell_1$ guarantee and a constant approximation factor
        can be solved
        \begin{itemize}
            \item with
            $$
                m = O(k \cdot (1 + \log_b (n / k)))
            $$
            measurements for $\Bc_{k, b}$,
            \item with
            $$
                m = O\left(k \cdot \frac{\log(n / k)}{\log \log(n / k)} \right)
            $$
            measurements for $\Tc_k$, provided that $k = \omega(\log n)$.
        \end{itemize}

    \end{theorem}

\section{Proof of Lemma~\ref{concentration}}
    \label{concentration_proof}
        To prove the lemma we need the following version of Chernoff bound~\cite{MR95}.
        \begin{theorem}[\cite{MR95}]
            \label{chernoff}
            Suppose that $X_1$, $X_2$, \ldots, $X_n$ are independent binary random variables.
            Denote $\mu = \Exp{}{X_1 + X_2 + \ldots + X_n}$. Then for any $\tau > 0$
            $$
                \Prb{}{X_1 + X_2 + \ldots + X_n \geq (1 + \tau) \mu} \leq
                \left(\frac{e^{\tau}}{(1 + \tau)^{(1 + \tau)}}\right)^{\mu}.
            $$
        \end{theorem}
        We can enumerate all $dt$ outgoing from $T$ edges arbitrarily.
        Then let us denote $C_u$ the event ``$u$-th of these edges 
        collides with $v$-th edge on the right side, where $v < u$''.

        We would like to upper bound the probability of the event
        ``at least $\eps dt$ of events $C_u$ happen''. As shown in~\cite{BMRV02}
        this probability is at most
        $$
            \Prb{}{B\left(dt, \frac{dt}{m}\right) \geq \eps dt}.
        $$
        Thus we can apply Theorem~\ref{chernoff} with $\tau = \eps m / dt - 1$
        and $\mu = (dt)^2 / m$.
        Then, the statement of the Lemma can be verified routinely.

\section{Proof of Theorem~\ref{model_upper}}

        If we combine Lemma~\ref{concentration}
        with the union bound we see that
        it is sufficient to prove that
        $$
            \forall 1 \leq t \leq l \quad \#(\Mc_{k}, t) \cdot \left(\delta \cdot
            \frac{\eps m}{dt}\right)^{-\eps dt} < \frac{1}{2l},
        $$
        and
        $$
            \forall l \leq t \leq k \quad \#(\Mc_{k}, t) \cdot \left(\delta \cdot
            \frac{\eps m}{dt}\right)^{-\eps dt} < \frac{1}{2k},
        $$
        provided that $m \geq C dk / \eps$ (here $C$ and $\delta$ are the
        constants from the statement of Lemma~\ref{concentration} and
        $l$ is the quantity from the statement of
        Theorem~\ref{model_upper}).
        Now let us plug~(\ref{block_estimates}).
        Namely, for $1 \leq t \leq l$ we use the second estimate, for
        $l \leq t \leq k$ we use the first estimate.
        Thus, it is left to choose $d$ such that we have
        \begin{eqnarray}
            \label{block_low_inequality}
            \forall 1 \leq t \leq l &\quad& \left(\frac{en}{t}\right)^t
            \left(\frac{C'k}{t}\right)^{-\eps dt} < \frac{1}{2l}\\
            \label{block_high_inequality}
            \forall l \leq t \leq k &\quad& |\Mc_k| \cdot
            \left(\frac{ek}{t}\right)^t
            \left(\frac{C'k}{t}\right)^{-\eps dt} < \frac{1}{2k},
        \end{eqnarray}
        for $C'$ being sufficiently large constant.
        It would allow us to set $m = O(dk / \eps)$.

        It is straightforward to check that, whenever $d > 1 / \eps$, the
        left-hand sides of~(\ref{block_low_inequality})
        and~(\ref{block_high_inequality}) are log-convex. Thus,
        to check~(\ref{block_low_inequality})
        and~(\ref{block_high_inequality}), it sufficient to check them
        for $t = 1, l, k$.

        Let us first handle~(\ref{block_low_inequality}).
        For it to hold for a given $1 \leq t \leq l$ it is enough
        to set
        $$
        d = O\left(\frac{1}{\eps \cdot \log(k / t)} \cdot
        \left(\frac{\log l}{t} + \log(n/t)\right)\right).
        $$
        For $t = 1$ this expression is $O(\eps^{-1} \cdot \log_k n)$,
        for $t = l$ it is
        \begin{equation}
            \label{d_final}
            O\left(\frac{\log(n / l)}{\eps \cdot \log(k / l)}\right).
        \end{equation}
        Clearly, the latter bound is always larger than the former
        one.

        Now let us handle~(\ref{block_high_inequality}).
        Similarly, for $l \leq t \leq k$ it is enough to set
        \begin{equation}
        \label{d_second}
        d = O\left(\frac{1}{\eps} \cdot
        \left(1 + \frac{\log |\Mc_k| + \log k}{t \cdot \log(k / t)}\right)\right)
        = O\left(\frac{1}{\eps} \cdot
        \left(1 + \frac{l \cdot \log(n / k) + \log k}{t \cdot \log(k / t)}\right)\right),
        \end{equation}
        where the second step 
        is due to the definition of $l$ ($l =
        \log |\Mc_k| / \log(n / k)$). 

        Now we will prove that for $t = l, k$ the right-hand side
        of~(\ref{d_second}) is at most~(\ref{d_final}).

        For $t = l$ we obviously have
        $$
            \frac{l \cdot \log(n / k)}{t \cdot \log(k/t)} =
            O\left(\frac{\log(n/l)}{\log(k/l)}\right).
        $$
        At the same time
        $$
            \frac{\log k}{t \cdot \log(k/t)} = 
            O\left(\frac{\log(n/l)}{\log(k/l)}\right),
        $$
        since
        $\log k=O(l \cdot \log(n/l))$. Indeed, if $l \geq \log n$,
        then we are done. Otherwise, $\log(n / l) = \Omega(\log n)
        = \Omega(\log k)$, and we are again done.

        Now let us handle $t = k$.
        Indeed,
        $$
        \frac{\log k}{t \cdot \log(k/t)} = O(1) 
        $$
        in this case.
        At the same time,
        $$
            \frac{l \cdot \log(n / k)}{t \cdot \log(k/t)} =
            O\left(\frac{\log(n/l)}{\log(k/l)}\right),
        $$
        since
        $l \cdot \log(k / l) = O(k)$.

        Overall, this argument shows that we can take
        $$
            d = O\left(\frac{\log(n / l)}{\eps \cdot \log(k / l)}\right)
        $$
        and
        $$
        m = O\left(\frac{k}{\eps^2}\cdot\frac{\log(n / l)}{\log(k / l)}\right).
        $$
\end{document}